\pdfoutput=1 
\documentclass{article}

\usepackage{arxiv}

\usepackage[utf8]{inputenc} % allow utf-8 input
\usepackage[T1]{fontenc}    % use 8-bit T1 fonts
\usepackage{hyperref}
\hypersetup{colorlinks,allcolors=black}

\usepackage{url}            % simple URL typesetting
\usepackage{booktabs}       % professional-quality tables
\usepackage{amsfonts}       % blackboard math symbols
\usepackage{nicefrac}       % compact symbols for 1/2, etc.
\usepackage{microtype}      % microtypography
\usepackage{lipsum}
\usepackage{graphicx}
\graphicspath{ {./images/} }

\usepackage{graphicx}
\usepackage{physics}
\usepackage{amsmath}
\usepackage{breqn} %***
\usepackage{xcolor} %***
\usepackage{siunitx} %***
\usepackage[british]{babel} %***
\usepackage{tabularx} %***
\usepackage{multirow}
\usepackage{algpseudocode}
\usepackage{algorithm}
\usepackage[utf8]{inputenc}
\DeclareUnicodeCharacter{2212}{-}
\usepackage{amsmath}

\usepackage{caption}
\usepackage{amsthm}
\usepackage{array}
\newcolumntype{P}[1]{>{\centering\arraybackslash}p{#1}}
\definecolor{maiblue}{rgb}{0, 0., 0.69}
\hypersetup{colorlinks=true, citecolor=maiblue, pdffitwindow=true, linkcolor=maiblue, urlcolor=maiblue}

\usepackage{hyperref}
\usepackage{color, colortbl}

\definecolor{Gray}{gray}{0.925}

\newtheorem{lemma}{Lemma}

\title{BILP-Q: Quantum Coalition Structure Generation}

\author{
 Supreeth Mysore Venkatesh \\
  Department of Mathematics and Computer Science\\
   University of Saarland\\
  66123 Saarbrucken, Germany\\
  \texttt{s8sumyso@stud.uni-saarland.de} \\
   \And
 Antonio Macaluso \\
  Agents and Simulated Reality Department \\
  German Research Center for Artificial Intelligence (DFKI)\\
  66123 Saarbrucken, Germany\\
  \texttt{antonio.macaluso@dfki.de} \\
  \And
 Matthias Klusch \\
  Agents and Simulated Reality Department \\
  German Research Center for Artificial Intelligence (DFKI)\\
  66123 Saarbrucken, Germany\\
  \texttt{matthias.klusch@dfki.de} \\
}

\begin{document}
\maketitle
\begin{abstract}
Quantum AI is an emerging field that uses quantum computing to solve typical complex problems in AI. 
In this work, we propose BILP-Q, the first-ever general quantum approach for solving the Coalition Structure Generation problem (CSGP), which is notably NP-hard. In particular, we reformulate the CSGP in terms of a Quadratic Binary Combinatorial Optimization (QUBO) problem 
to leverage existing quantum algorithms (e.g., QAOA) to obtain the best coalition structure. Thus, we perform a comparative analysis in terms of time complexity between the proposed quantum approach and the most popular classical baselines.
Furthermore, we consider standard benchmark distributions for coalition values to test the BILP-Q on small-scale experiments using the IBM Qiskit environment. 
Finally, since QUBO problems can be solved operating with quantum annealing, we run BILP-Q on medium-size problems using a real quantum annealer (D-Wave).
\end{abstract}

% keywords can be removed


\section{Introduction}

Quantum computation leverages quantum mechanics laws to endow quantum machines with tremendous computing power, enabling the solution of problems impossible to address with classical devices.
These premises are hugely appealing for many real-world applications, especially when coming to the adoption of quantum computing in the domain of Artificial Intelligence (AI). For this reason, Quantum AI is attracting ever-increasing attention from the academic and private sectors, even if its full potential is still to be understood. However, a topic that is rarely covered regards the adoption of quantum algorithms in Coalition Game Theory \cite{moura2018game}.

\subsection{Background} \label{sec:background}

%\section{Coalition Structure Generation}

The Coalition Structure Generation problem (CSGP) consists of the formation of coalitions by agents such that the social welfare is maximized. In practice,
%, a coalition structure (CS) is a set of agents partitioned into exhaustive and disjoint coalitions.
given a set of ~$n$~ agents ~$A=\{a_1,a_2,....a_n\}$~ and a characteristic function ~${v}:\mathcal{P}(A) \to \mathbb{R}$~, a coalition $C$ is a non-empty subset of ~$A$~. A Coalition Structure ($CS$) is a set of coalitions $\{C_1,C_2,...C_k\}$ such that ~$\bigcup_{i=1}^{k} C_{i} = A$~ and ~$C_i\cap C_j = \emptyset$~ for any ~$i,j \in \{1,2,...,k\}$~ and ~$i \neq j$~. The coalition value of a coalition structure is defined as ~${v}(CS) = \Sigma {v}(C)$~ for all ~$C \in CS$.
In a CSGP, given a set of agents and the characteristic function, the goal is to find the CS with the maximum coalition value. 

The CSGP can be also modeled as a Binary Integer Linear Programming (BILP) problem as follows \cite{rahwan2015coalition}. 
A $n \times (2^n-1)$ binary matrix $S$ is defined, where $n$ is the number of agents and $2^n-1$ are all possible coalitions of $n$ agents. The single entry $S_{i,j}$ is equal to $1$ if the agent $a_i$ belongs to the coalition  $C_j$ (i.e., $a_i \in C_j$), $0$ otherwise. %In this case, the solution $x=\{x_1,x_2,..,x_{2^n-1}\}$ is a $(2^n-1)$ binary string  which describes the best coalitions to create.
The BILP formulation consists of finding a $(2^n-1)$ binary string $x=\{x_1,x_2,..,x_{2^n-1}\}$ such that: %is a constraint optimization problem 
%where the objective is to maximize and the constraint to be satisfied can be defined using the following formulae:
\begin{equation}\label{eq: CSGP as BILP}
\text{Maximize}  \sum_{j=1}^{2^n-1} v(C_j) x_j 
\end{equation}
\begin{equation}\label{eq: constraint CSGP as BILP}
\text{subject to}  \sum_{j=1}^{2^n-1} S_{i,j} x_j=1 , \quad x_j \in \{0,1\} \\ 
\end{equation}
%\vspace{-.5em}
\begin{equation*}
%\text{Maximize}  \sum_{j=1}^{2^n-1} v(C_j) x_j \\
%\text{subject to}  \sum_{j=1}^{2^n-1} S_{i,j}.x_j=1 , \quad x_j \in \{0,1\} \\ 
\text{for} \ i=1,2,....,n  \quad \quad
\text{for} \ j=1,2,.....,(2^n-1)
\end{equation*}
 
where ${v}:\mathcal{P}(A) \to \mathbb{R}$ is the characteristic function of the game, $n$ is the number of input agents and $2^n-1$ is the cardinality of the power set $\mathcal{P}(A)$ of $n$ agents (empty set excluded), i.e., the set of all possible coalitions.

There are two broad classes of solutions for solving the CSGP. Anytime optimal algorithms (IP) \cite{rahwan2007anytime} generate an initial set of possible solutions within a bound from the optimal, and then improve the quality of these solutions. The downside of IP is that these algorithms might end up searching the entire space of all possible coalition structures, which translates in a worst-case time complexity of $\mathcal{O}(n^n)$. An alternative approach consists of using Improved Dynamic Programming (IDP) \cite{rahwan2008improved} which avoids the evaluation of all possible solutions without losing the guarantees of finding the optimal coalition structure. %These algorithms scale as $\mathcal{O}(3^n)$.
%, where $n$ is the number of agents in the coalition game. 
Importantly, the state-of-the-art solution for CSGP is represented BOSS algorithm \cite{changder2021boss} which combines IP and IDP, inheriting the worst-time complexity of IDP, which is $\mathcal{O}(3^n)$.

\subsection{Contribution}

This work proposes BILP-Q, the first general quantum approach for solving the CSGP using quantum computation. In practice, we consider the CSGP as BILP, reformulate it as Quadratic Unconstrained Binary Optimization (QUBO) problem, and leverage the QAOA \cite{farhi2014quantum} as a method to find the optimal coalition structure. As a further contribution, we analyze BILP-Q in terms of gate complexity as a function of the number of agents in the coalition game and compare it with IP and IDP-BOSS. 
Furthermore, we perform small-scale experiments using IBM Qiskit to show the effectiveness of the proposed approach. Finally, since QUBO problems can be solved operating with quantum annealing,
%as a heuristic solution, 
we run BILP-Q on medium-size problems using a real quantum annealer device (D-Wave).

\section{Related works}
Recently, a specific formulation of the CSGP for quantum annealing has been proposed \cite{leon2017multiagent} and further improved \cite{leon2019expressing}. In this case, the CSGP is expressed as an undirected weighted graph $G = (V, W)$, where the set of nodes $V$ corresponds to the set of agents and the edges $W: V \to V$ represent possible coalition structures. The weight assigned to each edge is given by the characteristic function $v(C_i) = \sum_{(i,j) \in C_i}{w_{ij}}$. Therefore, the whole problem is reformulated as a minimization problem and mapped directly into the topology of a quantum annealer. % such that it is possible to find the best coalition structure using quantum annealing. 

We identify several drawbacks to this approach. First, the ability to solve a specific problem instance depends on the specific topology of the quantum annealer in use. This means that if the graph generated by a given problem does not fit the connectivity of the qubits in the available quantum device, it is impossible to run the algorithm.  
Second, the problem formulation requires two extra parameters,  $c_{max}$ and \textit{crossing number},  that control the maximum number of possible coalitions and the minimum number of intersecting edges, respectively. In practice, these parameters are unlikely to be known in real case scenarios. Third, the use cases explored consider only superadditive games \cite{rahwan2015coalition}.
 
Thus, the proposed formulation in \cite{leon2017multiagent, leon2019expressing} is not general-purpose, since it requires relevant apriori information and specific hardware architecture to solve the CSGP. 

\section{BILP-Q: Quantum Algorithmic Solution for CSGP}
 
 In this section, we propose  BILP-Q, a general QUBO formulation for CSGP, completely independent of the specific problem instance,  which allows the adoption of gate-based quantum algorithms (e.g., QAOA) and real quantum annealers (D-Wave) as methods to find the optimal coalition structure. 
 Furthermore, we analyze the computational complexity of BILP-Q that uses QAOA and compare it with classical baselines IP and IDP-BOSS.

\subsection{QUBO Formulation for CSGP}

 %In this section we provide a general formulation for the CSGP as Quadratic Unconstrained Binary Optimization problem, such that it can be solved using QAOA and Quantum Annealing. 
 Starting from the BILP formulation in Eq. \eqref{eq: CSGP as BILP}, we rewrite the CSGP in terms of quadratic objective function, in matrix form, as follows:
  \begin{align}\label{eq: QUBO generic}
 \text{Maximize} & \quad  f^'(\boldsymbol{x}) = \boldsymbol{x}_{1 \times (2^n-1) }^t\boldsymbol{C}_{(2^n-1) \times (2^n-1) }\boldsymbol{x}_{(2^n-1) \times 1 } \\
 \text{subject to} & \quad \boldsymbol{S}_{n \times (2^n-1)}\boldsymbol{x}_{(2^n-1) \times 1 } = \boldsymbol{b}_{n \times 1 }, \label{eq: constraint quadratic}
 \end{align}

 where $\boldsymbol{x}$ is a $2^n-1$ binary vector, $\boldsymbol{C}$ is a $(2^n-1) \times (2^n-1)$ diagonal matrix whose entries are given by the characteristic function $v(C_j)$,  $\boldsymbol{b}$ is an all-ones vector and $\boldsymbol{S}$ is a binary matrix whose rows and columns represent the agents and all possible coalitions, respectively.
 The transformation of Eq. \eqref{eq: CSGP as BILP} as quadratic function is possible since $\boldsymbol{x}$ is a binary vector and for any element $x_j \in \boldsymbol{x}$,  $x^2_j = x_j$. 
 Importantly, the matrix $\boldsymbol{S}$ is highly sparse since each agent (row) belongs, at most, to half of all possible coalitions (columns)\footnote{Given a set $A$ of $n$ agents and the corresponding power set $\mathcal{P}(A)$ of size $2^n$, the number of subsets whose the generic agent $a_i$ belongs (for  $i \in \{1, \dots n \}$) is $2^{n−1}$ since this is equivalent as finding the number of subsets of a set with size $n-1$.}.
  Notice that, it is possible to express the maximization of $f^'(\boldsymbol{x})$ in Eq. \eqref{eq: QUBO generic} as the minimization of  $f(x)$ which is defined as $f(x) = -f^'(\boldsymbol{x})$.
%, this allows to leverage the relation between Ising model and QUBO problems to solve CSGP using quantum computation. 
%From now on, we consider the minimazation of $f(\boldsymbol{x}) = -f^'(\boldsymbol{x})$.
% Furthermore, one can thing of arbitrarily force as many zero's as possible in or
 
 We embed the constrains of Eq. \eqref{eq: QUBO generic} into the objective function $f(\boldsymbol{x})$, adding a Lagrangian penalty term to shift from a constrained optimization problem to an unconstrained problem \cite{kochenberger2004unified}, providing a complete QUBO formulation for CSGP:
 %\noindent 
\begin{flalign}\label{eq: QUBO formulation CSGP}
        f(\boldsymbol{x}) & =  \boldsymbol{x}^t \boldsymbol{C} \boldsymbol{x} + \lambda \left(\boldsymbol{S}\boldsymbol{x}-\boldsymbol{b}\right)^t \left(\boldsymbol{S}\boldsymbol{x}-\boldsymbol{b}\right) \nonumber \\
        & =  \boldsymbol{x}^t \boldsymbol{C} \boldsymbol{x} + \lambda \left(\boldsymbol{x}^t \boldsymbol{S^t}\boldsymbol{S} \boldsymbol{x} - \boldsymbol{x}^t \boldsymbol{S}^t \boldsymbol{b} - \boldsymbol{b}^t \boldsymbol{S} \boldsymbol{x} +  \boldsymbol{b}^t\boldsymbol{b} \right) 
        \nonumber \\ & 
        = \boldsymbol{x}^t \boldsymbol{C} \boldsymbol{x} + \boldsymbol{x}^t \boldsymbol{D} \boldsymbol{x} + c  
        \nonumber \\ & 
        =\boldsymbol{x}^t \boldsymbol{Q} \boldsymbol{x} + c
\end{flalign}
where $Q \in \mathbb{R}^{(2^n - 1) \cross (2^n - 1)}$ is a symmetric matrix, $c$ is an additive constant that does not affect the optimization process and $\lambda$ is an arbitrarily large positive number, known as \textit{penalty parameter}, that allows to penalize all those solutions that are possible in the new QUBO formulation but explicitly forbidden in the previous formulations \eqref{eq: constraint CSGP as BILP} and \eqref{eq: constraint quadratic}. 
%\textcolor{red}{where P is a penalty parameter that allows to  the matrix $D$ and the additive constant $c$ are the results from the matrix multiplication, $Q \in \mathbb{R}^{(2^n - 1) \cross (2^n - 1)}$ is the QUBO matrix which consists the linear coefficients as the diagonal elements and quadratic terms as the non-diagonal elements. By dropping the additive constant $c$, we arrive at the optimization problem:}
Thus, the original CSGP can be reformulated as a minimization of the QUBO problem in Equation \eqref{eq: QUBO formulation CSGP} as follows:
\begin{align} \label{eq: MAX Qubo for CSGP}
    \text{Minimize} \ f(\boldsymbol{x}) & = \boldsymbol{x}^t \boldsymbol{Q} \boldsymbol{x} = \sum_{i=1}^{2^n-1}c_{i} x_i +\sum_{1 \leq i < j< 2^n-1}^{2^n-1} q_{ij} x_i x_j
\end{align}
where the coefficients $c_i$ are proportional to the elements of the matrix $\boldsymbol{C}$ by the factor $\lambda$, and the binary string solution provides the encoding (matching to the columns of the matrix $\boldsymbol{S}$) to the best coalition structure (i.e., maximum coalition value). 
The new formulation for CSGP depends completely on the structure of the matrix $ \boldsymbol{Q}$, whose off-diagonal elements depend, in turn, on the original matrix $ \boldsymbol{S}$. 
As already mentioned in Sec. \ref{sec:background}, the sparsity of $\boldsymbol{S}$ (i.e., the number of non-zero elements) is strictly lower than its total entries, and this is also true for the upper/lower off-diagonal elements of $\boldsymbol{Q}$. In addition, in case of Constrained Coalition Formation \cite{bistaffa2017algorithms} the number of possible coalitions is further reduced and many of the entries in $ \boldsymbol{S}$ are forced to be equal $0$. In these cases, the matrix  $\boldsymbol{S}$ can be even more sparse (as well as $\boldsymbol{Q}$).
In the next section, we show that having an arbitrarily sparse matrix $\boldsymbol{Q}$ directly affect the gate complexity of QAOA when it is used 
%for solving CSGP, 
in the context of BILP-Q.

% the idea of using highly sparse matrices to 

%\textcolor{red}{Mention quantum Spline}
%  QUBO problems allows are notably NP-complete and the advantage... 
% Scrivere della minimizzazione prima della equazione

\subsection{Solving CSGP using QAOA }
 
 Solving a QUBO problem is equivalent to finding the ground state of an Ising model Hamiltonian 
 %(NP-complete \cite{barahona1982computational}) 
 where the binary variables $x_i$ are replaced by spin variables $z_i=2x_1-1$, $z_i \in \{-1, +1\}$. In particular, denoting as $\mathcal{I}$ the set of pair-wise interactions between spins (i.e., the number of non-zero off-diagonal elements in $\boldsymbol{Q}$), one can formulate the energy $E(\boldsymbol{z})$ of the spin system as follows:
 \begin{align} \label{eq: Ising model for CSGP}
    E(\boldsymbol{z}) = \sum_{i=1}^{2^n-1} h_{i} z_i + \sum_{(i,j) \in \mathcal{I}}^{i} J_{i,j} z_i z_j,
 \end{align}
where the ground state of the Hamiltonian corresponds to the solution of the original QUBO problem. %, that in our case is equivalent of finding the binary encoding of best coalition structure.
Once the QUBO is formulated in terms of Ising Hamiltonian, one can leverage QAOA \cite{farhi2014quantum} to find the optimal solution. In practice, when using a $p$-level QAOA acting on a $2^n-1$ qubits ($n$ number of agents), the quantum system is evolved with a cost Hamiltonian $H_C$ and a mixing Hamiltonian $H_B$ $p$ times. % to find the optimal binary encoding of the best coalition structure. ù
These two Hamiltonians are constructed as follows. 
$H_B$ is fixed and results from the sum of all single qubits $\sigma^x$ operators $H_B=\sum_{j=1}^{2^n-1} \sigma_j^x$.
$H_C$, instead, is realized by replacing the binary variable with Pauli-$Z$ operations as $\frac{1}{2} \sum_{(i,j) \in \mathcal{I}} (1-\sigma_i^z \sigma_j^z)$. 
Starting from $H_C$ and $H_B$, two sets of parametrized unitary matrices 
$U(H_C, \gamma_j) = e^{-i \gamma_j H_C}$ and $U(H_B, \beta_j) = e^{-i \beta_j H_B}$, for $j = 1 \dots p$, are defined, which depend on the two sets of parameters $\boldsymbol{\gamma} = (\gamma_1, \dots, \gamma_p )$ and $\boldsymbol{\beta} = (\beta_1, \dots, \beta_p )$. 
In terms of quantum circuit, each of the Pauli-$X$ in $H_B$ is implemented with a single $R_X(\beta)$ rotation gate, while each of the two-qubit interactions in $H_C$ is implemented with two CNOT gates and a local $R_Z(\beta)$ single-qubit gate \cite{crooks2018performance}.
The standard approach for QAOA consists of generating the uniform superposition of all states in the computational basis:     $\ket{s} =  H^{\otimes 2^n-1} \ket{0}^{\otimes 2^n-1}$.
Then the unitaries $U(H_C, \gamma_j)$ and $U(H_B, \beta_j)$ are applied iteratively $p$ times to produce the variational quantum state: %for $p \geq 1$ and $2p$ angles $(\gamma_1, \dots , \gamma_p)$ and $(\beta_1, \dots , \beta_p)$ as follows:
\begin{align}
    \hspace{-1em}\ket{\boldsymbol{\beta}, \boldsymbol{\gamma}} =  e^{-i \beta_p H_B} e^{-i \gamma_{p}  H_C} \ \cdots \ e^{-i \beta_1 H_B} e^{-i \gamma_1 H_C}\ket{s} =  U(\boldsymbol{\beta}, \boldsymbol{\gamma})\ket{s}
\end{align}
The third step consists of performing expectation measurement to all the qubits: 
\begin{align}
    F_p(\boldsymbol{\beta}, \boldsymbol{\gamma}) =  \bra{\boldsymbol{\beta}, \boldsymbol{\gamma}} M \ket{\boldsymbol{\beta}, \boldsymbol{\gamma}},
\end{align}
where $M$ is a generic measurement operator.
Once the quantum state has been measured, the two sets of parameters $\boldsymbol{\beta}$ and $\boldsymbol{\gamma}$ are updated using classical optimization and the whole process is repeated multiple times to find 
the value of $F_p(\boldsymbol{\beta}, \boldsymbol{\gamma})$ for the near-optimal values $(\boldsymbol{\beta}^*, \boldsymbol{\gamma}^*)$ that minimizes $ E(\boldsymbol{z})$. For more technical details see \cite{zahedinejad2017combinatorial, majumdar2021optimizing}.

As a consequence of the cost Hamiltonian $H_C$ being constructed on the interactions $\mathcal{I}$ of the Ising model \eqref{eq: Ising model for CSGP}, and the equivalence with QUBO formulation \eqref{eq: MAX Qubo for CSGP}, the sparsity of the matrix $\boldsymbol{Q}$ directly affects the number of gates in QAOA needed to solve the CSGP. In fact, the QUBO matrix $\boldsymbol{Q}$ of BILP-Q formulation is sparse by nature since the number of non-zero off-diagonal elements in the matrix $\boldsymbol{S}$ is generally high, and it can be arbitrarily increased to forbid specific coalitions in case of constrained coalition formation problems \cite{bistaffa2017algorithms}. This translates in a lower computational complexity (in terms of quantum gates) with respect to a generic QUBO formulation. 
As a result, there are cases where BILP-Q can outperform even the state-of-the-art solution when comparing the time complexity of the classical baselines with the gate complexity of BILP-Q as a function of the number of agents (Section \ref{sec:performance}).

Notably, the idea of identifying specific classes of problems where the sparsity affects the computational complexity of a quantum algorithm has shown good results in the context of fault-tolerant quantum machine learning \cite{macaluso2020quantum}.

\begin{figure}[h]
    \centering
    \includegraphics[scale=.7]{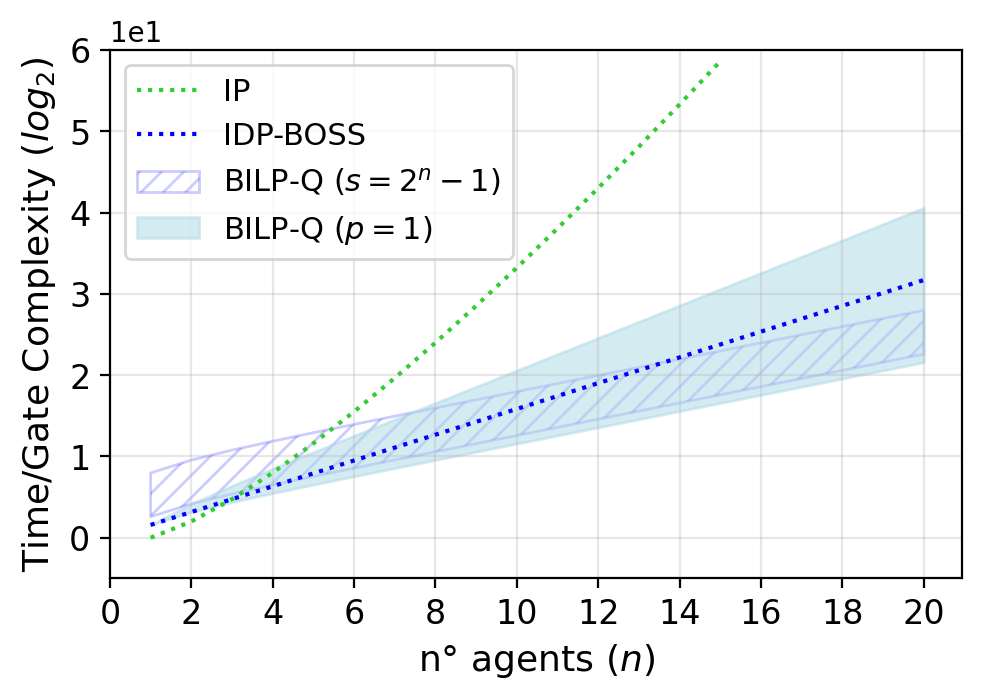}%[width=\linewidth, height = 5cm]
    \caption{Cost complexity as a function of the number of agents $n$. The green curve represents IP $(\mathcal{O}(n^n))$, while the blue one is referred to the time complexity of the IDP-BOSS $(\mathcal{O}(3^n))$. The light blue shaded area and the dashed area illustrate the performance of BILP-Q while varying $s$ and $p$ respectively.}
    \label{fig:complexity}
\end{figure}

\subsection{Performance analysis}\label{sec:performance}

In this section, we analyze the computational complexity of BILP-Q for a generic CSGP instance. In particular, we express 
%space complexity (i.e., number of required qubits) and 
the gate complexity of BILP-Q as a function of the number of agents in the coalition game. Furthermore, we compare BILP-Q with the time complexity of popular classical solutions for CSGP.
However, when comparing classical and quantum algorithms, it is important to consider that quantum computing introduces a new complexity class, the \textit{Bounded-error Quantum Polynomial time}, representing the class of problems solvable in polynomial time by an innately probabilistic quantum Turing machine. % \cite{nielsen2002quantum}. 
% \begin{lemma} \label{lemma:bilpq_number _of_qubits}
% \textit{For a CSG problem of $n$ agents, BILP-Q that makes use of QAOA algorithm, requires $2^{{(n)}}-1$ qubits to compute the binary encoding of the optimal coalition structure.}
% \end{lemma}
% %\vspace{-2em}
% \begin{proof}
% According to the QUBO formulation provided in Eq. \ref{eq: Qubo for CSGP}, the final solution of CSGP corresponds to a $2^n-1$-dimensional binary string whose length is equivalent to the number of possible coalitions, and the specific binary values indicate if a certain coalition has to be build or not, according to the encoding strategy provided by $S$. QAOA algorithms initialise a quantum state where all possible solution for a QUBO optimization problem are in equal superposition, then according to the Ising Hamiltonian provided Eq. ... and the parameter $p$ train the quantum circuit to output the binary string which provide the solution of the original problem with high probability. Thus, BILP-Q that makes use of QAOA to solve the QUBO problem of size $2^n-1$ output  the right solution $2^n-1$ qubits are required.
% \end{proof}
\begin{lemma} \label{lemma:bilpq_number _of_qubits}
\textit{
% For a $n$-agent Coalition Structure Generation problem, 
% %of with a number $\kappa$ of interaction in the matrix $\boldsymbol{Q}$, 
% BILP-Q, that uses a $p$-layered QAOA algorithm, requires 
% $\mathcal{O} \Big( 3 \times ( 2^n + p s -1) \Big)$
% single and/or two qubit gates to compute the optimal coalition structure, where $s$ is the number of non-zero element. 
For a $n$-agent CSGP, BILP-Q, that uses a $p$-layered QAOA, requires $\mathcal{O} \Big((2^n-1)( 2p+1) + 3ps \Big)$
single and/or two-qubit gates to compute the optimal coalition structure. The $s$ parameter is the number of non-zero elements in the lower/upper triangular matrix constructed from $\boldsymbol{Q}$. % of BILP-Q formulation.
}
%The number $k$ is the number of non-zero elements in the matrix $\boldsymbol{Q}$  }
\end{lemma}
%\vspace{-3em}
\begin{proof}
We consider the two Hamiltonians $H_C$ and $H_B$ generated according to the $\boldsymbol{Q}$ matrix of the BILP-Q formulation. 
Starting from $2^n-1$ qubits as input, the first step of QAOA generates an equal superposition of $2^{2^n-1}$ possible states through the use of $2^n-1$ Hadamard gates. Then, for each interaction in $H_C$, three gates (two CNOT gates and a local single-qubit $R_Z$ gate) are employed, plus an additional $R_Z$ applied to each qubit, even with a Ising model with no interactions. We define as $s$ the total number of interactions, which is also equivalent to number of non-zero elements in the upper/lower triangular matrix $\boldsymbol{Q}$ (i.e., the cardinality of the set $\mathcal{I}$ in Eq \eqref{eq: Ising model for CSGP}). The size of the matrix $\boldsymbol{Q}$ is $(2^n-1)\times(2^n-1)$, hence there are $2^{n-1}(2^{n}-3) + 1$ 
%\footnote{$\frac{(2^n-1)\times(2^{n}-2)}{2}$} 
off-diagonal elements. Furthermore, as already mentioned in Sec. \ref{sec:background}, the matrix $\boldsymbol{S}$ is heavily sparse by construction and can be further modified by restricting the possible coalitions arbitrarily. Thus, the value $s$ can potentially vary from $2^n-1$ (a single non-zero off-diagonal element per row) to $2^{n-1}(2^{n}-3) + 1$ (total number of the off-diagonal elements in $\boldsymbol{Q}$ which is however impossible to achieve in case of BILP-Q). Finally, $H_B$ is implemented using $2^n-1$ Pauli-$X$ single-qubit rotation gates $R_X$. The two Hamiltonians are iteratively applied $p$ times. 
Thus, the total number of single or two-qubit gates of BILP-Q is the following: 
\begin{equation*}
\underbrace{2^n-1}_{\text{Hadamard}} + %\underbrace{p}_{\text{layers QAOA}} 
p \times 
(\underbrace{3\times s + 2^n-1}_{\text{$H_C$}} + \underbrace{2^n-1}_{\text{$H_B$}} ).
\end{equation*}
It is easy to show that this number is equal to  $(2^n-1)( 2p+1) + 3ps$.
\end{proof}
Figure \ref{fig:complexity} illustrates a theoretical comparison of BILP-Q computational cost with respect to the classical solutions.
The blue curve is referred to the time complexity of IDP-BOSS, while the light blue shaded area illustrates the cost (gate complexity) of BILP-Q as the sparsity of $\boldsymbol{Q}$ varies from $2^n-1$ to $2^{n-1}(2^{n}-3) + 1$. BILP-Q outperforms the IDP-BOSS as the number of interactions decreases. Also, varying $p$ from $1$ to $50$ (with fixed $s=2^n-1$), for large problem instances ($n\geq 14$), BILP-Q outperforms IDP-BOSS even with large $p=50$.  
For medium-size problem instances ($n\geq 6$), BILP-Q outperforms IP in all the cases.

Importantly, the comparison between gate complexity and (classical) time complexity is valid only assuming to efficiently train the QAOA using classical optimization, a topic which is still open in quantum computing \cite{guerreschi2019qaoa}.

\section{Experiments}

To test the BILP-Q we performed small-scale experiments on most of the distributions for coalition values usually employed to test classical algorithms. %such as BOSS (for more details about the distributions see \cite{changder2021boss}).
%is a game with competition between groups of players ("coalitions") due to the possibility of external enforcement of cooperative behavior (e.g. through contract law).
In particular, different cooperative games (or coalition games)  are generated by sampling the coalition values $v(\cdot)$ from the following probability distributions: %used for testing IDP and IP:
agent-based uniform (ABU), agent-based normal (ABN), modified uniform (MU), Normal (N), and Single Valuable Agent with beta (SVA-$\beta$). 
%\cite{rahwan2012hybrid, michalak2016hybrid, adams2010approximate, larson2000anytime, biswas2012coalition}. 
Additionally, the distributions used to benchmark the BOSS algorithm \cite{changder2021boss} are tested: Weibull (W), Rayleigh (R), Weighted random with Chisquare (WRC), F-distribution, Laplace (LAP). %For more details see \cite{changder2021boss}.

\begin{table}[ht]
\begin{center}
%\addtolength{\tabcolsep}{-1pt} % adjust to fit
    \begin{tabular}[\linewidth]{c|ccc|ccc}%{c S[table-format=3.6] S[table-format=3.2] S[table-format=3.2] S[table-format=3.2]}%{lllll}
    \toprule
     & \multicolumn{3}{c}{\textbf{2 agents}} & \multicolumn{3}{c}{\textbf{3 agents}} \\
    \cmidrule(lr){2-4}\cmidrule(lr){5-7}
    \textbf{Distr.} & \textbf{p}  & \textbf{BILP-Q}  & \textbf{BILP}  & \textbf{p} &  \textbf{BILP-Q}  & \textbf{BILP} \\
            \toprule
\rowcolor{Gray}  ABN        & 1 &   0.033 &        0.062 &   17 &   3.133 &        0.517 \\
                 ABU        & 1 &   0.034 &        0.062 &   10 &   1.642 &        0.624 \\
\rowcolor{Gray}  F          & 3 &   0.086 &        0.063 &    7 &   0.910 &        0.588 \\
                 Laplace    & 1 &   0.120 &        0.065 &   11 &   2.026 &        0.608 \\
\rowcolor{Gray}  MU         & 1 &   0.029 &        0.060 &   10 &   2.360 &        0.641 \\
                 Normal     & 1 &   0.043 &        0.102 &    2 &   0.490 &        0.576 \\
\rowcolor{Gray}  Rayleigh   & 5 &   0.173 &        0.062 &   12 &   2.491 &        0.603 \\
                 SVA-$\beta$& 2 &   0.053 &        0.080 &    6 &   0.883 &        0.567 \\
\rowcolor{Gray}  WRC        & 2 &   0.055 &        0.068 &    6 &   1.582 &        0.529 \\
                 Weibull    & 1 &   0.042 &        0.070 &    6 &   1.499 &        0.601 \\
                 [0.1cm]
                \bottomrule
    \end{tabular}%
        \vspace{1em}
        {\caption{Results of BILP-Q. The table shows the optimal parameter $p$ for each problem instance, the time to train the QAOA in the BILP-Q formulation, and the time required to solve the CSGP classically after formulating it as a BILP (Eq. \eqref{eq: CSGP as BILP},\eqref{eq: constraint CSGP as BILP}). The times of BILP and BILP-Q are reported in milliseconds $(ms)$. \vspace{-2em}}
            \label{tab:QAOA_results}}
    \end{center}
\end{table}

Within BILP-Q which adopts the QAOA as quantum solution, the number of agents of the game is restricted only to $2$ and $3$, due to the limitation in training large instances of this quantum algorithm using quantum simulation \cite{willsch2020benchmarking}. 
Specifically, the QAOA is trained using IBM Qiskit on the QASM simulator, a backend that simulates the execution of a quantum algorithm in a fault-tolerant setting. The number of measurements for each run of the quantum circuits is fixed to $1024$. The execution times for each problem instance are reported in Table \ref{tab:QAOA_results}.

We can observe that the optimal parameter $p$ is strictly related to the specific problem instance to be solved and, as expected, it has a relevant impact on the time required to train the QAOA, since it increases the depth of the quantum circuit proportionally to $p$. In fact, calculating the Pearson's correlation coefficient between the optimal parameter $p$ and the time required to find the optimal solution with BILP-Q (Table \ref{tab:QAOA_results}), a value of $0.96\%$ is obtained. This means that as long as the value of $p$ grows, the time to find the optimal solution increases linearly to it.
Also, on average, the optimal parameter $p$ is significantly higher as the number of agents increases. This might be due to the number of qubits required $(2^n-1)$ which increases from $3$ to $7$ when solving the the problem with $2$ and $3$ agents respectively.  
Importantly, since QAOA in the context of BILP-Q has been trained on a quantum simulator, the execution times reported in Table \ref{tab:QAOA_results} are just an indication on how BILP-Q would work on gate-based quantum computation on different problem instances.

Furthermore, since quantum annealing can be adopted as a heuristic solution for QUBO problems, BILP-Q formulation has been tested using a real quantum annealer (D-Wave) \cite{cohen2014d} for problem instances with a number of agents from $2$ to $7$. Results are shown in Figure \ref{fig:dwave_results}. 

We can observe an increasing trend in the quantum annealer time as the number of agents increases. In particular, this tendency seems to be exponential with respect to $n$ for all the distributions, as testified by the light green curve which describes the function $f(n) = 2^n+60$. 
This is an interesting finding since it represents a better results than the state-of-the-art IDP-BOSS which scales as $\mathcal{O}(3^n)$. Nonetheless, it is necessary to perform further experiments to verify whether these times hold even with a larger number of agents.
\begin{figure}[ht]
\centering
\includegraphics[scale=.5]{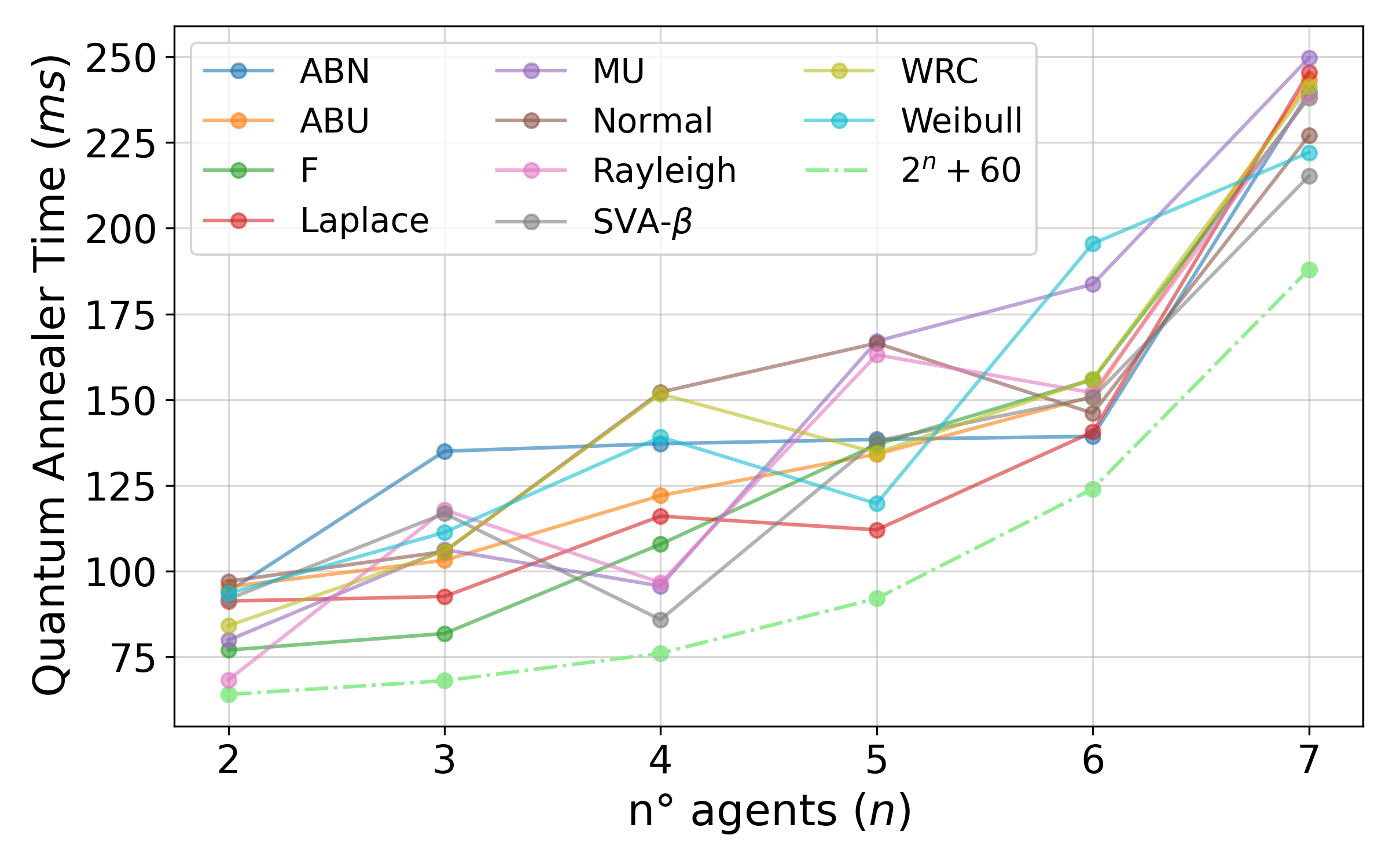} %[width=\linewidth, height = 6cm]{images/QA.png} 
\caption{Results of BILP-Q using a D-Wave quantum annealer.
%- \lq\lq exponential behavior as $2^n$\rq\rq
}
\label{fig:dwave_results}
\end{figure}
% \vspace{-1em}

In general, for both quantum implementations of BILP-Q, the specific coalition values seem not to be particularly relevant for the time required to get the solution. As expected, the number of agents $n$ and the parameter $p$ (for QAOA) have a huge impact. However, these results do not represent an exhaustive evaluation of BILP-Q but only a proof-of-concept on different possible coalition games. Thus, further experiments need to be performed to test in practice its potential advantage over classical solutions.

\section{Conclusion}

In this work, we proposed BILP-Q, a quantum algorithmic approach for the Coalition Structure Generation problem. In particular, we presented a general QUBO formulation for CSGP that can be solved using different types of quantum devices. Furthermore, we compared the computational (gate) complexity of BILP-Q against the most adopted classical counterparts, showing that the CSGP may be an ideal setting for leveraging quantum computation in AI.

Also, we demonstrated the adoption of BILP-Q for small-size problems using gate-based quantum simulation and tested it on a real D-Wave quantum annealer when increasing the number of agents of the coalition game up to $7$.

Future studies will be dedicated to identifying specific cases for CSGP where it is possible to leverage the advantages of BILP-Q in real-world applications, especially those related to Constrained Coalition Formation \cite{bistaffa2017algorithms}. Finally, the implementation of QAOA for larger problem instances will be investigated using ad-hoc inizialization procedures \cite{egger2021warm} as well as custom optimization strategies for quantum annealing.

\section*{Code Availability}

All code to generate the data, figures, analyses, as well as, additional technical details on the experiments are publicly available at \href{https://github.com/supreethmv/BILP-Q}{\textcolor{blue}{https://github.com/supreethmv/BILP-Q}}.

\section*{Acknowledgments}

This work has been funded by the German Ministry for Education and Research (BMB+F) in the project QAI2-QAICO under grant 13N15586.

\bibliographystyle{unsrt}  
\bibliography{references}  %%% Remove comment to use the external .bib file (using bibtex).

% %%% Comment out this section when you \bibliography{references} is enabled.

% George Kour and Raid Saabne.
% \newblock Real-time segmentation of on-line handwritten arabic script.
% \newblock In {\em Frontiers in Handwriting Recognition (ICFHR), 2014 14th
%   International Conference on}, pages 417--422. IEEE, 2014.

% George Kour and Raid Saabne.
% \newblock Fast classification of handwritten on-line arabic characters.
% \newblock In {\em Soft Computing and Pattern Recognition (SoCPaR), 2014 6th
%   International Conference of}, pages 312--318. IEEE, 2014.

% \bibitem{hadash2018estimate}
% Guy Hadash, Einat Kermany, Boaz Carmeli, Ofer Lavi, George Kour, and Alon
% \newblock Estimate and replace: A novel approach to integrating deep neural
%   networks with existing applications.
% \newblock {\em arXiv preprint arXiv:1804.09028}, 2018.

% \end{thebibliography}

\end{document}